\documentclass[runningheads]{llncs}
\usepackage{graphicx}
\usepackage{subfigure}
\usepackage{epstopdf}
\usepackage{cite}
\usepackage{amsmath,amssymb,amsfonts}
\begin{document}
\title{Sync or Fork: Node-Level Synchronization Analysis of Blockchain}
%
%
\author{Qin Hu\inst{1} \and
Minghui Xu\inst{2} \and
Shengling Wang\inst{3(*)} \and Shaoyong Guo\inst{4}}
\authorrunning{Q. Hu et al.}
%
\institute{
Indiana University - Purdue University Indianapolis, USA
\and
The George Washington University, USA \and
Beijing Normal University, China \and
Beijing University of Posts and Telecommunications, China \\
* Corresponding author \\
Email: \email{qinhu@iu.edu, mhxu@gwu.edu, wangshengling@bnu.edu.cn, syguo@bupt.edu.cn}}
\maketitle              
\begin{abstract}
As the cornerstone of blockchain, block synchronization plays a vital role in maintaining the security. Without full blockchain synchronization, unexpected forks will emerge and thus providing a breeding ground for various malicious attacks. The state-of-the-art works mainly study the relationship between the propagation time and blockchain security at the systematic level, neglecting the fine-grained impact of peering nodes in blockchain networks. To conduct a node-level synchronization analysis, we take advantage of the large deviation theory and game theory to study the pull-based propagation from a microscopic perspective. We examine the blockchain synchronization in a bidirectional manner via investigating the impact of full nodes as responders and that of partial nodes as requesters. Based on that, we further reveal the most efficient path to speed up synchronization from full nodes and design the best synchronization request scheme based on the concept of correlated equilibrium for partial nodes. Extensive experimental results demonstrate the effectiveness of our analysis.

\end{abstract}

\section{Introduction}
Since the appearance of Bitcoin \cite{nakamoto2019bitcoin}, 
 cryptocurrency as the killer application of blockchain piques substantial attention from the whole society to the underlying distributed ledger technology. Research on blockchain from all walks of life indicates its great potential and versatility. It is reported that the blockchain market size over the globe reaches \$3 billion in 2020 and is expected to surge to \$39.7 billion by 2025 \cite{market}.

As the infrastructure of blockchain systems, the peer-to-peer network consisting of peering nodes supports the most important operations of information dissemination and exchange, including both control and data messages. 
To maintain the consistent recognition of the main chain, 
the synchronization of newly generated blocks among all nodes becomes extremely important. Otherwise, unexpected forks will emerge, which might be exploited by malicious clients to achieve various attacks, such as double spending and selfish mining, and can even lead to the breakdown of a blockchain system. 

To enable the synchronization of blockchain, there exist five types of block propagation mechanisms \cite{gervais2016security}, i.e., advertisement, header sending, unsolicited push, relay network, and push-advertisement hybrid. 
Focusing on the behaviors of nodes propagating block information, 
we can summarize them as 
pull-based and push-based. 
In the pull-based propagation, nodes with timely information of the blockchain, termed as full nodes, respond to the requests of updating block information from neighboring nodes, named as partial nodes, which can achieve block synchronization  cost-efficiently in an on-demand manner.  
While in the push-based one, any node receiving the newly generated block automatically pushes this piece of information to neighbors, which can synchronize the blockchain network quickly but will cause unnecessary communication among nodes. 
Other works about blockchain synchronization mainly study the relationship between the propagation time and blockchain security at the systematic level \cite{decker2013information,sompolinsky2015secure,gobel2016bitcoin,kang2018incentivizing}, neglecting the fine-grained impact of peering nodes in blockchain networks.

In this paper, we study the blockchain synchronization in pull-based propagation from a microscopic perspective, using the large deviation theory and game theory to investigate different roles of peering nodes in synchronizing block information. 
This suggests the feature of \textit{node-level} analysis of this work. Besides, our research is \textit{bidirectional}, which captures the feature of impacts on block synchronization from two main types of nodes in the blockchain, i.e., the full node as the responder and the partial node as the requester. 
Specifically, we reveal  clues about three critical questions: \textit{How will the full node's response capability affect the synchronization? How to efficiently reduce its negative effect on synchronization? And how should the partial node to actively achieve the synchronization?}

In summary, our contributions in this work include the following three aspects:
\begin{itemize}
\item The impact of full nodes on synchronization is quantitatively characterized by the concept of response failure rate, which straightforwardly uncovers the synchronization probability of connected partial nodes.
\item The negative impact of full nodes on synchronization can be fast eliminated via increasing the decay speed of the response failure rate, and the derived expression of the decay speed indicates that enlarging the response capacity related parameter is more efficiently than improving the response rate. This paves a clear path to facilitate synchronization from full nodes.
\item The optimal synchronization scheme for the partial node is established based on the concept of correlated equilibrium, where a Node Synchronization (NS) problem is formulated to guarantee 
that the partial node can get synchronized without unnecessary cost or redundant response from full nodes.
\end{itemize}

The remaining of this paper is organized as follows. In section \ref{sec:related}, we investigate the most related work on blockchain synchronization. Then we introduce the node-level synchronization model for both the full node and the partial node in Section \ref{sec:formulation}, where the full node's response capability is further analyzed in Section \ref{sec:analysis} while the best synchronization request mechanism for the partial node is presented in Section \ref{sec:solution}. All theoretical analysis are evaluated in Section \ref{sec:experiment}. And finally, we conclude the whole paper in Section \ref{sec:conclusion}.

\section{Related Work}\label{sec:related}
Similar to traditional distributed systems, there are three levels of synchrony of blockchain networks,  
namely synchronous, partially synchronous, and asynchronous. As the representative blockchain application, Bitcoin whitepaper \cite{nakamoto2019bitcoin} provides an initial analysis on its security against forks and double-spending attacks with an oversimplified model. Since 2015, Bitcoin consensus algorithm has been thoroughly investigated considering three levels of synchrony \cite{garay2015bitcoin}\cite{pass2017analysis}. Garay \textit{et al.} \cite{garay2015bitcoin} formalize the Bitcoin consensus within a fully synchronous network. 
Persistence and liveliness are proved to be guaranteed hinging on the synchronous setting. Pass \textit{et al.} \cite{pass2017analysis} show that Bitcoin consensus satisfies consistency and liveliness in a partially synchronous network, but consistency cannot be satisfied in an asynchronous network. 

As the most critical factor affecting blockchain synchronization, the propagation time of control messages and data messages 
is investigated to reveal how it affects blockchain security against various attacks, such as forks, double spending, and selfish mining, and how to mitigate the corresponding vulnerabilities. 
The propagation time 
is shown as the primary cause for blockchain forks \cite{decker2013information}. In response, researchers propose three methods to speed up propagation: minimizing verification, pipelining block propagation, and increasing connectivity. 
Sompolinsky and Zohar \cite{sompolinsky2015secure} study the relation between higher transaction rate and the vulnerability to double-spending attacks, which shows that increasing block size and block generation rate can improve the throughput, but will increase the propagation time so that even weaker attackers can launch double-spending attacks. 
Besides, the selfish mining is investigated in a realistic setting where propagation time is taken into account \cite{gobel2016bitcoin}, indicating that it becomes easier with increasing propagation delay. 
For PoS-based consensus, Kang \textit{et al.} \cite{kang2018incentivizing} propose a Stackelberg game based incentive mechanism to encourage miners to propagate blocks, enabling lower propagation delay and higher security level. 

For the propagation mechanism in blockchain networks, five popular categories are summarized in \cite{gervais2016security}, 
including advertisement, header sending, unsolicited push, relay network, and push-advertisement hybrid. Early on, the advertisement-based propagation is adopted by Bitcoin, which has a two-round message exchange procedure. Afterward, Bitcoin resorts to the header propagation to avoid using $inv$ messages. In unsolicited push propagation, miners directly broadcast newly-mined blocks. 
The relay network, adopted by FIBRE \cite{BitcoinRelay}, BloXroute\cite{klarman2018bloxroute}, and Geeqchain \cite{conley2018geeq}, is to distribute relay nodes globally to which miners can connect to and exchange information at a high speed. However, relay nodes are criticized for introducing centralization to blockchain. Ethereum adopts push and advertisement hybrid propagation by which a node can automatically push messages to $\sqrt{n}$ nodes and advertises messages to neighboring nodes simultaneously \cite{wust2016ethereum}. 

In summary, existing works about blockchain synchronization focus on macroscopically investigating blockchain protocol to figure out the relationship between propagation time and security or propose new propagation mechanisms. However, in this paper, we study the blockchain synchronization from a microscopic and node-level perspective, using the large deviation theory and game theory to depict blockchain nodes precisely and investigate how nodes' capability affect synchronization.

\section{System Model}\label{sec:formulation}

%
%


In this paper, 
we assume that full nodes are homogeneous in terms of information request and response performance. Thus, we can shed light on the synchronization status of the whole blockchain system via studying the response capability of any specific full node. 
And all partial nodes are also assumed to be similar in terms of interacting with full nodes to get synchronized. As full nodes and partial nodes play different roles in blockchain synchronization, we introduce their models separately in the following.

\subsection{Response Model of the Full Node}
Considering that the requests of updating block information from partial nodes arrive at the full node randomly, we assume that this stochastic event is a Poisson process with arrival rate $\lambda$, which is inspired by the typical model of packet arrival process in communication networks \cite{cao2001internet}. It usually takes some time for the full node to respond and send out the latest block information since the node might be busy on handling other tasks, which can also be assumed as a Poisson process with response rate $\mu$.
To guarantee that the full node can finish responding to the requests from partial nodes most of the time, we assume $\mu > \lambda$. However, even with this condition, there might still exist some cases where the full node fails to respond.

To investigate this issue, we define the number of synchronization requests arrived at the full node and that the node can respond during time period $(t-1,t)$ as $a_t$ and $r_t$, respectively, where $t \in \mathbb{N}^*$. Then we can describe the request queue at the full node as 
\begin{equation}
Q_t = (Q_{t-1}+a_t-r_t)^+, \nonumber
\end{equation}
where $(\cdot)^+$ denotes the positive part of the inside expression.

Next, we focus on the cumulative arrival and response process, denoted as $A_t=a_1+\cdots + a_t$ and $R_t=r_1 + \cdots + r_t$, respectively. Thus, the length of the request queue until time $t$ at the full node, defined as $L_t$, will be 
\begin{equation}\label{eq:Lt}
L_t =  A_t - R_t.
\end{equation}
Generally speaking, since $\mu > \lambda$, one may expect that $L_t$ would be negative, making it pointless with the definition of queue length. However, due to the randomness of the arrival and response process, the queue length can become positive, which may even overwhelm the response capability of the full node, leading to the failure of responding synchronization requests. 
To prepare for the worst case of response failure in blockchain, we focus on the maximum possible queue length at the full node when $t\rightarrow \infty$, which is defined as $\mathcal{L}=\sup_{t>0} L_t$, and further investigate 
the possibility of the request queue being over-length, i.e., $\mathcal{L}>\Gamma$, where $\Gamma$ is defined as follows:

\begin{definition}[Response capacity]\label{def:capacity}
The response capacity $\Gamma$ of the full node is the longest queue of synchronization requests that it can process without any failure.
\end{definition}


According to Definition \ref{def:capacity}, we can know that if $\mathcal{L} \leq\Gamma$, the fulll node can handle all synchronization requests successfully. 
But if $\mathcal{L}>\Gamma$, the request queue is too long for the full node to handle, which will make the partial nodes sending block synchronization requests fail to achieve the distributed consistency. 
To analyze this important event, we introduce the following definition:

\begin{definition}[Response failure rate]
The response failure rate is the probability that the longest synchronization request queue arrived at the full node, i.e., $\mathcal{L}$, exceeds its response capacity $\Gamma$, denoted as $P(\mathcal{L}>\Gamma)$.
\end{definition}

With the help of $P(\mathcal{L}>\Gamma)$, we can capture the full node's failure of responding to block synchronization requests in a quantitative manner, which provides us a more straightforward clue about the synchronization status of the neighboring partial nodes. 
Based on this index, we can make adjustment or countermeasure in time to avoid unpredictable loss brought by the asynchronous blockchain information among partial nodes, which will be analyzed in Section \ref{sec:analysis}. 

\subsection{Synchronization Model of the Partial Node}
We assume that the number of partial nodes in a blockchain network is $N$, and each of them has direct access to multiple full nodes to obtain block synchronization information. Specifically, for any partial node, we denote the set of full nodes it has direct connections as $\mathcal{M} = \{M_i\},~i \in \{1,\cdots,m\}$, where $m \in \mathbb{N}^*$ is the number of full nodes. 
And the above-defined response failure rate of these full nodes can be denoted as $P_i,~i \in \{1,\cdots,m\}$.

For a cautious partial node, it may send the synchronization request to all connected full nodes so as to obtain a higher successful synchronization probability. Thus, the synchronization failure event can only happen to this partial node when all full nodes failed to respond with the latest block information, which means that the synchronization failure probability is $\prod_{i=1}^m P_i$. And accordingly, the successful synchronization probability of this partial node is $1-\prod_{i=1}^m P_i$.

While in a more general case, a normal partial node might need to seriously consider where to send the synchronization request. First, sending the request costs communication resource, and thus generously sending the request to all available full nodes can bring too much burden on the resource consumption for the partial node. What's more, with the assumption that all full nodes have the same new information of the blockchain, it would be enough for the partial node to receive at least one response and thus other redundant responses become a waste. With this in mind, we can see that wisely sending the synchronization request is vital for the partial node, which will be elaborated in Section \ref{sec:solution}.

\section{Response Failure Analysis}\label{sec:analysis}
In this section, we analyze the response failure rate $P(\mathcal{L}>\Gamma)$ in detail. We first focus on the derivation of its decay speed $I(x)$, based on which two critical factors impacting the systematic response are discussed.

We first let $\Gamma = lx$ with $x>0$. Then according to the Cram$\acute{\mathrm{e}}$r's  theorem \cite{ganesh2004big}, for large $l$, there exists $P(\mathcal{L}>\Gamma) = P(\mathcal{L}>lx) \approx exp(-l I(x))$, which indicates that the probability of $\mathcal{L}>lx$ will decay exponentially with the rate $I(x)$ when $l \rightarrow \infty$. 
In detail, we have
\begin{equation}
\lim_{l\rightarrow \infty} \frac{1}{l} \log P(\mathcal{L}>lx) = -I(x),
\end{equation}
where $I(x)$ is the rate function with the following expression 
\begin{align}\label{eq:Ix}
I(x) = & \inf_{t>0} t \Phi^* (\frac{x}{t}).
\end{align}
According to the large deviation theory and the calculation process in \cite{wang2019corking}, we can have the expression of $I(x)$ as: 
\begin{equation}\label{eq:final_Ix}
I(x) = x \ln \frac{\mu}{\lambda}.
\end{equation}

As we mentioned earlier, $I(x)$ reveals the decay speed of response failure rate $P(\mathcal{L}>\Gamma)$. In other words, the larger $I(x)$, the sharper decrease of $P(\mathcal{L}>\Gamma)$, and thus the more successful the block synchronization for the requested partial nodes. With this in mind, we desire to enlarge $I(x)$ as much as possible. 
On one hand, from the above expression of $I(x)$, one can tell that it is linearly increasing with the response capacity related parameter $x$ when the ratio of the response rate $\mu$ to the arrival rate $\lambda$ is fixed. On the other hand, if $x$ is given, we can see $I(x)$ is logarithmically correlated to $\frac{\mu}{\lambda}$. 
Therefore, theoretically speaking, increasing $x$ is more effective to improve $I(x)$ than increasing $\mu/\lambda$, which will be numerically analyzed in Section \ref{sec:experiment}.

Considering that $x$ and $\frac{\mu}{\lambda}$ are two main factors impacting the value of $I(x)$, we study them further in the following. As $x$ is based on the response capacity $\Gamma$ and the arrival rate $\lambda$ is a system-wide parameter which cannot be  adjusted, we mainly focus on $\Gamma$ and $\mu$ since they are more controllable from the perspective of the full node. 
In the following, we investigate how to set $\Gamma$ and $\mu$ to meet some specific system-performance requirements on response failure rate. To this end, we first denote a \textit{response failure tolerance degree} as $\epsilon \in (0,1]$, which acts as the constraint for the failure rate $P(\mathcal{L}>\Gamma)$. And then we introduce the following two definitions.

\begin{definition}[Effective response capacity]\label{def:effective_syn}
The effective response capacity $\Gamma^*(\epsilon)$ is the minimum capacity that the full node needs to provide to enforce that the response failure rate will never greater than $\epsilon$, i.e.,
\begin{equation}
\Gamma^*(\epsilon) = \min \{\Gamma: P(\mathcal{L}>\Gamma) \leq \epsilon\}. \nonumber
\end{equation}
\end{definition}

\begin{definition}[Effective response rate]\label{def:effective_rate}
The effective response rate $\mu^*(\epsilon)$ is the minimum response rate requirement for the full node to guarantee that the response failure rate will never greater than $\epsilon$, i.e.,
\begin{equation}
\mu^*(\epsilon) = \min \{\mu: P(\mathcal{L}>\Gamma) \leq \epsilon\}. \nonumber
\end{equation}
\end{definition}

Further, we have the following theorems to present the specific results of $\Gamma^*(\epsilon)$ and $\mu^*(\epsilon)$.

\begin{theorem}\label{thrm:gamma}
For $\epsilon \in (0,1]$ and $\mu > \lambda$, we can calculate $\Gamma^*(\epsilon) $ as:
\begin{equation}
\Gamma^*(\epsilon)  = - \frac{\ln \epsilon}{\ln \frac{\mu}{\lambda}}. \nonumber
\end{equation} 
\end{theorem}
\begin{proof}
As we mentioned at the beginning of this section, for $l \rightarrow \infty$, we have $P(\mathcal{L}>\Gamma) \approx e^{-lI(x)}$. Then it comes to $e^{-lI(x)}\leq \epsilon$ according to Definition \ref{def:effective_syn}. 
Besides, based on \eqref{eq:final_Ix} and $\Gamma = lx$, we can prove that the value of $\Gamma^*(\epsilon)$ is $- \frac{\ln \epsilon}{\ln \frac{\mu}{\lambda}}$.
\end{proof}

\begin{theorem}
For $\epsilon \in (0,1]$ and $\mu > \lambda$, we can calculate $\mu^*(\epsilon) $ as:
\begin{equation}
\mu^*(\epsilon)  = \lambda e^{- \frac{\ln \epsilon}{\Gamma}}. \nonumber
\end{equation} 
\end{theorem}
\begin{proof}
Similar to the proof of Theorem \ref{thrm:gamma}, due to $P(\mathcal{L}>\Gamma) \approx e^{-lI(x)} \leq \epsilon$, we can have $lx \ln \frac{\mu}{\lambda} \geq -\ln \epsilon$, which leads to the result of $\mu^*(\epsilon)$.
\end{proof}

\section{Correlated Equilibrium based Node Synchronization Mechanism}\label{sec:solution}
As mentioned earlier, the synchronization of one certain partial node is collectively completed by the surrounding full nodes, which heavily depends on how many of them the partial node requests. In fact, each full node has a particular response capability with respect to the synchronization request, which is well captured by the response failure tolerance degree introduced in the above section, and it takes some cost for the partial node to send the synchronization request to a specific full node. For a reasonable and intelligent partial node, it is essential to work out an efficient and effective strategy to select the subset of full nodes as synchronization request targets. 
In other words, \textit{given different $\epsilon_i~(i \in \{1,\cdots,m\})$ of all connected full nodes, how should the partial node make decisions on whether to send the blockchain synchronization request to each of them?}

To solve this problem, we first define that the decision strategy of the partial node is $\mathbf{p}=(p_1,\cdots,p_m)$ with $p_i\in \{0,1\}$, where 0 (or 1) denotes not sending (or sending) the synchronization request to the full node $M_i$. From the perspective of the partial node, the ultimate goal of this decision is to guarantee that it can obtain the up-to-date information of the main chain from at least one full node. Thus, the profit of deciding whether to send the request to one specific full node $M_i$ is jointly affected by the decisions of sending to other full nodes, which can be defined as 
\begin{align*}
\phi_i (\mathbf{p}) = \frac{p_i(1-\epsilon_i)}{\sum_{j=1}^m p_j(1-\epsilon_i)}.
\end{align*}
Note that in the case of $\mathbf{p}=\mathbf{0}$, we define $\phi_i (\mathbf{p}) = 0$.

With $C_i$ denoting the cost of sending the request to $M_i$,  we can define the utility of this decision as
\begin{align}\label{eq:utility}
U_i (\mathbf{p})  = \alpha_i \phi_i	(\mathbf{p}) - p_i C_i,
\end{align}
where $\alpha_i > 0$ is a scalar parameter.

On one hand, as a utility-driven decision maker, the partial node desires to obtain an optimal utility for each individual decision about one specific full node, which is collectively affected by the decision vector $\mathbf{p}$ about all full nodes and can be described by the following game-theoretic concept named  \textit{correlated equilibrium}.

\begin{definition}[Correlated equilibrium]
Denote the strategy space as $\mathcal{V} = \{0,1\}$ with the size of $V=2$ and a probability distribution over the space $\mathcal{V}^m$ as ${G(\mathbf{p})}$. Then $G(\mathbf{p})$  is a correlated equilibrium if and only if $G(\mathbf{p})$ makes that for any decision $p_i,p'_i\in \mathcal{V}$, there exists
\begin{equation}
\sum_{\mathbf{p}_{-i}\in \mathcal{V}^{m-1}} G(p_i,\mathbf{p}_{-i}) \Big( U_i(p_i,\mathbf{p}_{-i}) - U_i(p'_i,\mathbf{p}_{-i}) \Big) \geq 0, \nonumber
\end{equation}
where $\mathbf{p}_{-i}=(p_1,\cdots,p_{i-1},p_{i+1},\cdots,p_n)$ denotes other decisions except for $p_i$. 
\end{definition}

The above definition implies that under the correlated equilibrium $G(\mathbf{p})$, there is no motivation for the partial node to deviate from the strategy $p_i$ about sending the request to $M_i$ given other strategies $\mathbf{p}_{-i}$. In other words, the partial node can only obtain the maximized utility with respect to the individual decision via selecting $p_i$ according to the decision vector $\mathbf{p}$ sampled from the correlated equilibrium  $G(\mathbf{p})$. It is obvious that there may exist various correlated equilibria  meeting the above-defined constraint.

On the other hand, the partial node cares about the overall utility of all decisions about all surrounding full nodes since it reflects the general synchronization status of this partial node, which can be calculated as $\sum_{\mathbf{p}\in \mathcal{V}^{m}} G(\mathbf{p}) \sum_{i=1}^m U_i(\mathbf{p})$. Therefore, we can obtain the best correlated equilibrium for the partial node considering the global optimization goal, which is summarized as the following Node Synchronization (NS) problem.

\textbf{NS Problem:}
\begin{align}
\max:~ & \sum_{\mathbf{p}\in \mathcal{V}^{m}} G(\mathbf{p}) \sum_{i=1}^m U_i(\mathbf{p}) \label{eq:objective} \\ 
\mathrm{s.t.}:~ & G(\mathbf{p})\geq 0, ~\forall \mathbf{p} \in \mathcal{V}^m, \label{eq:bigzero}\\
 &\sum_{\mathbf{p}\in \mathcal{V}^m} G(\mathbf{p})=1, \label{eq:equalone}\\
\sum_{\mathbf{p}_{-i}\in \mathcal{V}^{m-1}} G(p_i,\mathbf{p}_{-i}) \Big( & U_i(p_i,\mathbf{p}_{-i}) - U_i(p'_i,\mathbf{p}_{-i}) \Big) \geq 0,  \forall p_i,p'_i \in \mathcal{V}. \label{eq:correlated}
\end{align}

Obviously, the above NS problem is an optimization problem with respect to the variable $G(\mathbf{p})$, where the optimization object \eqref{eq:objective} is to maximize the overall expected utility for all decisions, constraint \eqref{eq:bigzero} is a natural requirement for the probability distribution, constraint \eqref{eq:equalone} refers to that the sum of all probability distribution is 1, and the last one \eqref{eq:correlated} is directly obtained from the definition of correlated equilibrium to achieve individual utility maximization. 
Besides, via scrutinizing the NS problem, one can find that it is exactly a linear programming problem with respect to the probability distribution $G(\mathbf{p})$. In fact, there exist a lot of efficient algorithms to solve the linear programming problem with polynomial time complexity, such as interior point and simplex-based algorithms.

\section{Experimental Evaluation}\label{sec:experiment}
In this section, we first numerically analyze the key factor impacting the response failure rate $P(\mathcal{L}>\Gamma)$ , i.e., the decay speed $I(x)$. 
 Further, the proposed correlated equilibrium based node synchronization mechanism is validated to demonstrate its effectiveness. 
Specifically, all experiments are carried out using a laptop running with 2.7 GHz Dual-Core Intel Core i5 processor and 8 GB memory. And for the sake of statistical confidence, we report average values of all experimental results via repeating each experiment for 20 times.

\subsection{Numerical Analysis of Response Failure Rate}
We first plot $I(x)$ changing with the response capacity related parameter $x$ and the response rate $\mu$ in Fig. \ref{fig:decay}. In particular, we use the difference between $\mu$ and $\lambda$, i.e., $\mu - \lambda$, to capture the impact of $\frac{\mu}{\lambda}$ in \eqref{eq:final_Ix} on $I(x)$ for easy understanding. Specifically, we set $x\in [0,1]$, $\lambda = 3$ and $\mu - \lambda \in [0,10]$. 

\begin{figure}[h]
\centering
\subfigure[]{
\begin{minipage}{0.45\linewidth}
\includegraphics[width=1\textwidth]{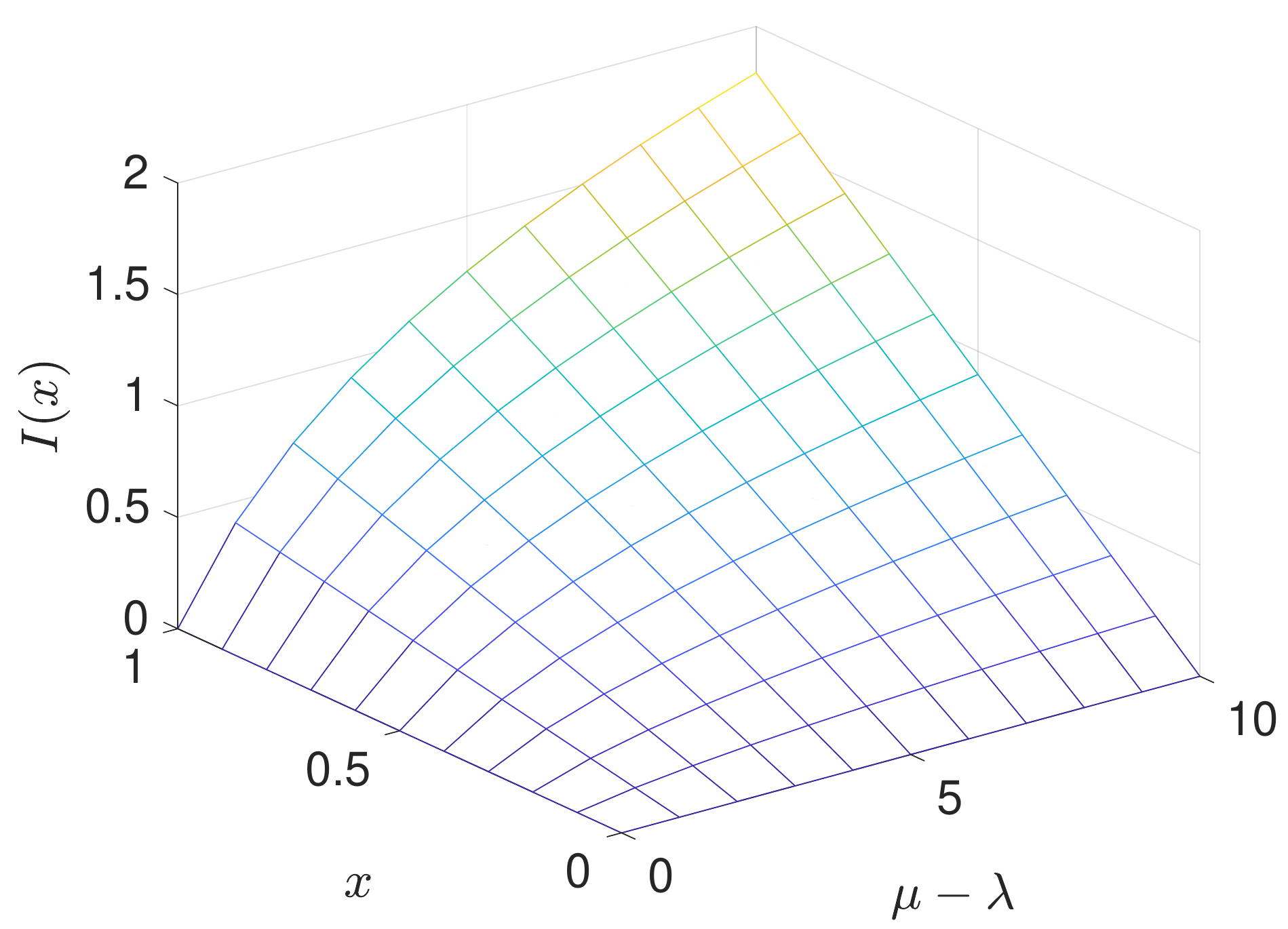}
\end{minipage}
}
\subfigure[]{
\begin{minipage}{0.23\linewidth}
\includegraphics[width=1\textwidth]{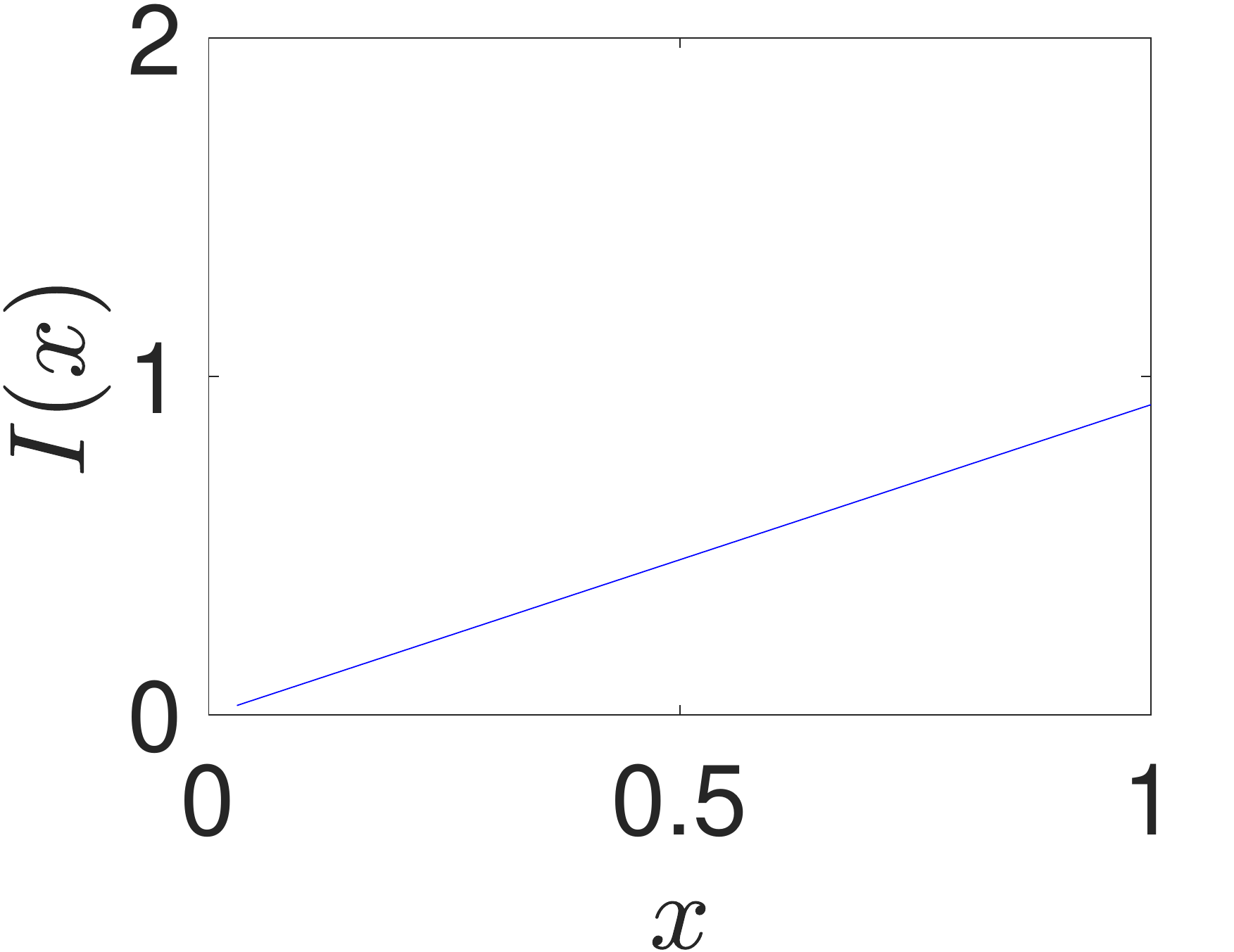}\\
\includegraphics[width=1\textwidth]{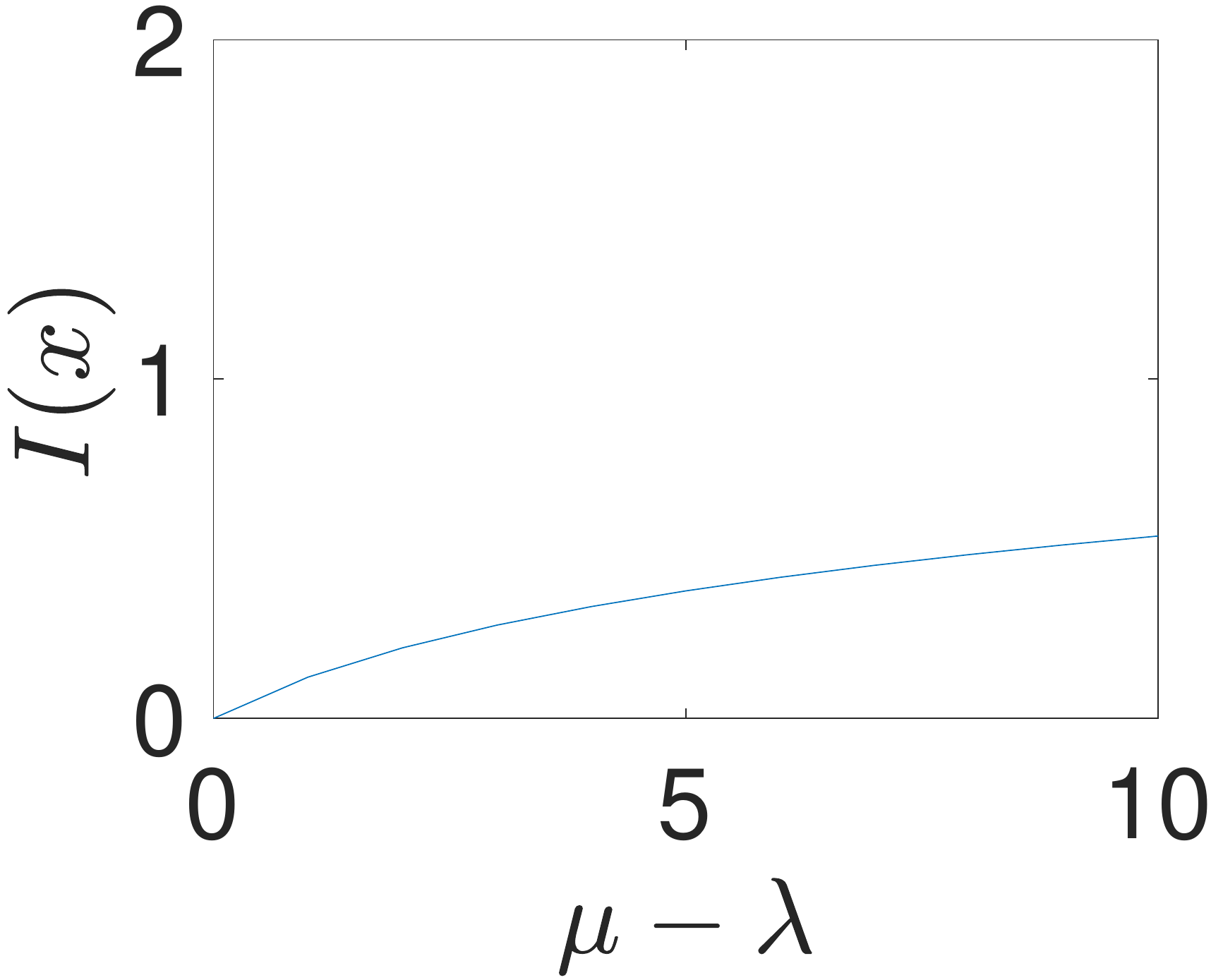}
\end{minipage}

\begin{minipage}{0.23\linewidth}
\includegraphics[width=1\textwidth]{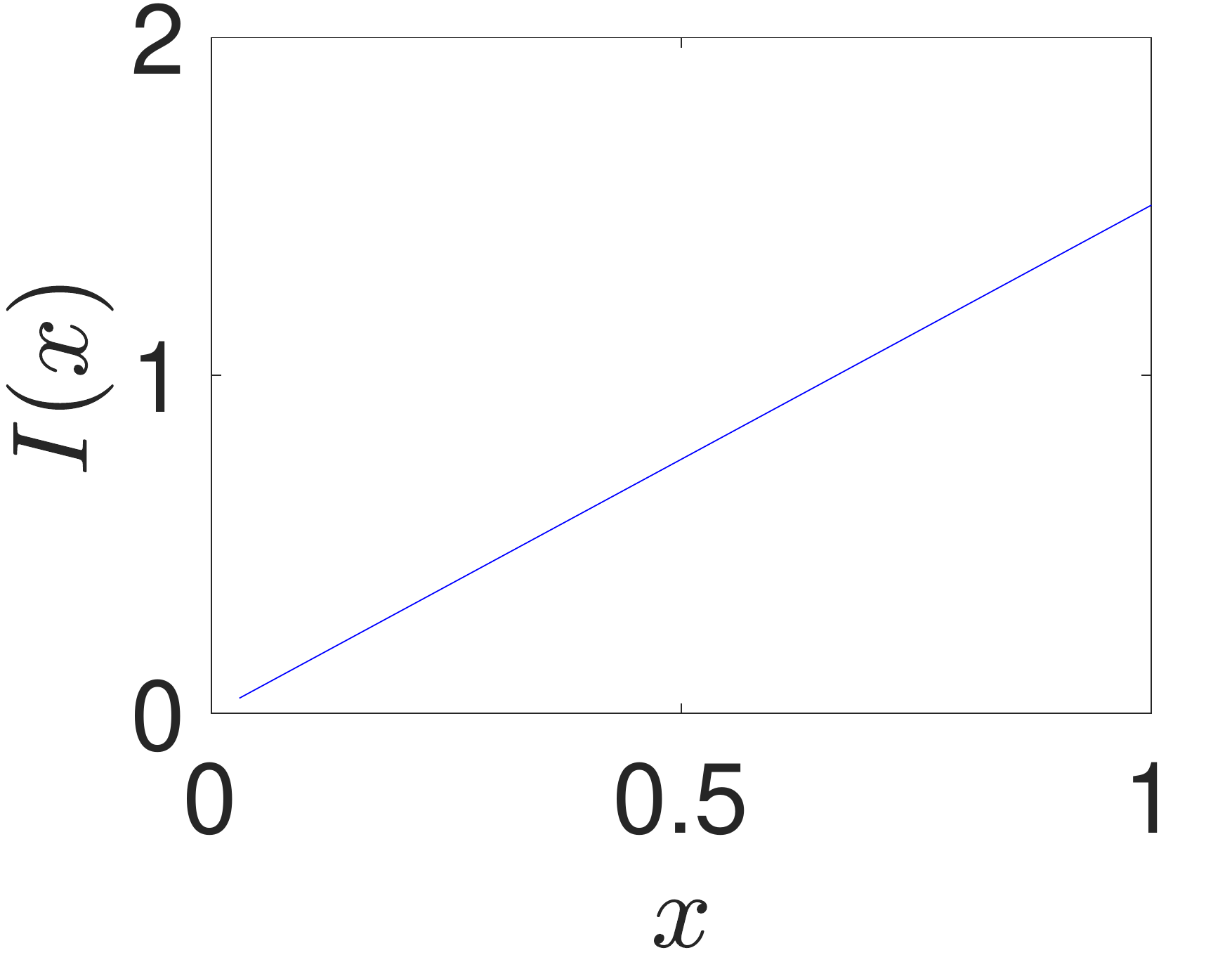}\\
\includegraphics[width=1\textwidth]{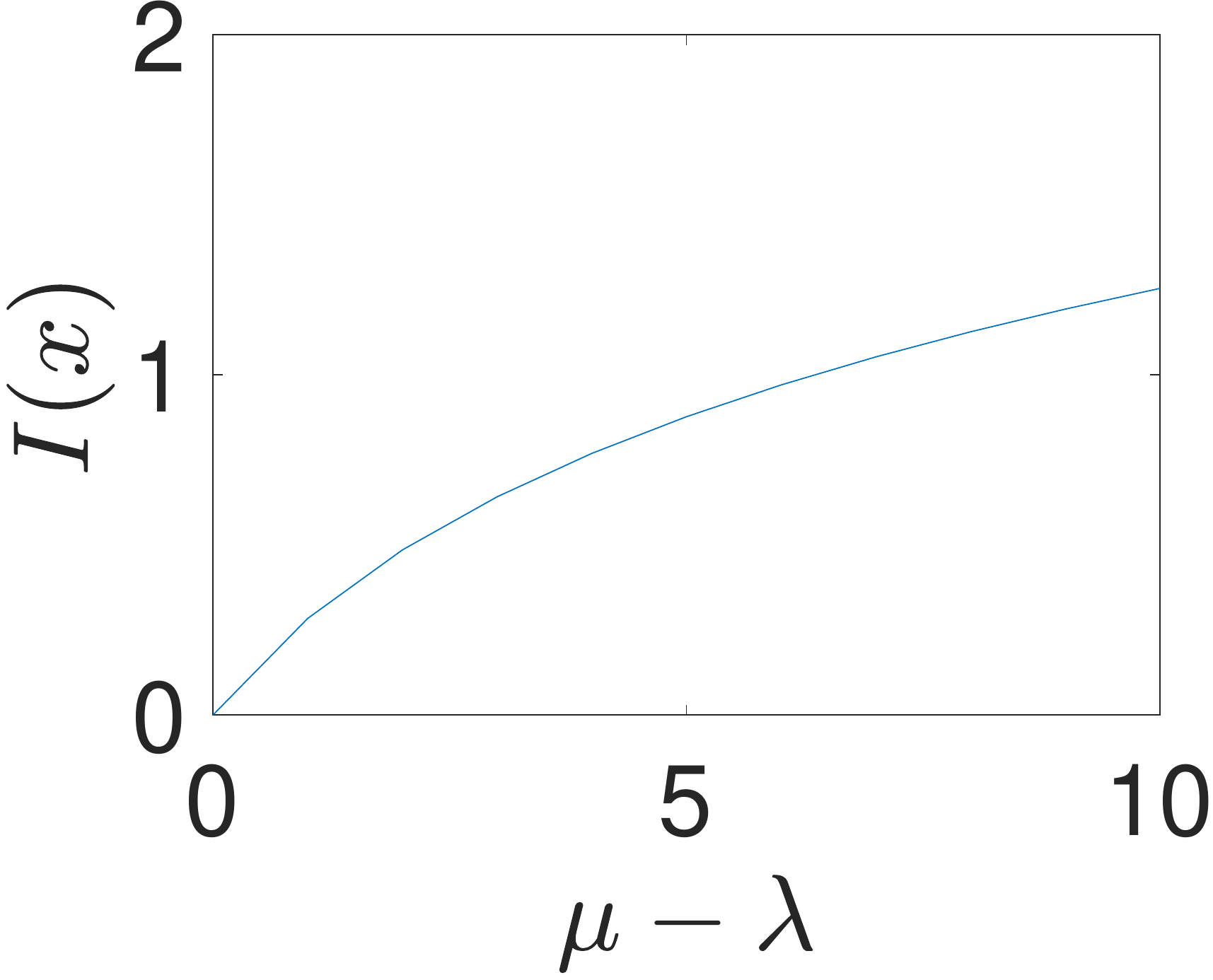}
\end{minipage}
}
\centering
\caption{Decay speed of the synchronization failure rate.}
\label{fig:decay}
\end{figure}

It can be seen that $I(x)$ increases with the larger $x$ and $\mu - \lambda$, which means that we can achieve a higher decay speed for the response failure rate via either improving the response capacity $\Gamma=lx$ or enhancing the full node's response rate $\mu$ given a specific arrival rate $\lambda$. Besides, via comparing the first and second lines of subfigures in Fig. \ref{fig:decay}(b), one can tell that the decay speed has different changing trends with respect to $x$ and $\mu - \lambda$, where the increasing $x$ can lead to linear change while the increase of $\mu - \lambda$ can only bring logarithmic variation. Thus, we can conclude that raising the response capacity can achieve a lower response failure rate more efficiently.

\subsection{Evaluation of Node Synchronization Mechanism}
Next, we explore the effectiveness of our proposed node synchronization mechanism in Section \ref{sec:solution}. In detail, we take the case of $m=8$ as an example and focus on the decision of sending the synchronization request to the full node $M_1$ who has a varying response failure tolerance degree $\epsilon_1 \in (0,1)$ with an interval of 0.1. Other parameters are set as $\alpha_i = 10, C_i = 5$. The request sending decisions are reported in Fig. \ref{fig:decision} with two representative cases, where the response failure tolerance degrees of all other full nodes, i.e., $M_2$ to $M_m$, are the same and fixed as $\epsilon_{-1}=0.2$ and 0.8. 
It is obvious that the request decision vectors in two cases are very different. In the case of $\epsilon_{-1}=0.2$ in Fig. \ref{fig:decision}(a), 
$p_1$ is 1 until $\epsilon_1$ is larger than others, which means that sending the request to $M_i$ is a good choice until its response failure rate is higher than others. And similarly, when $\epsilon_{-1}=0.8$ as shown in Fig. \ref{fig:decision}(b), $p_1$ keeps to be 1 except for $\epsilon_1=0.9$ which is larger than response failure rates of other full nodes.

\begin{figure}[h]
\centering
\subfigure[$\mathbf{\epsilon}_{-1} = 0.2$.]{
\begin{minipage}[t]{0.45\linewidth}
\centering
\includegraphics[width=1\textwidth]{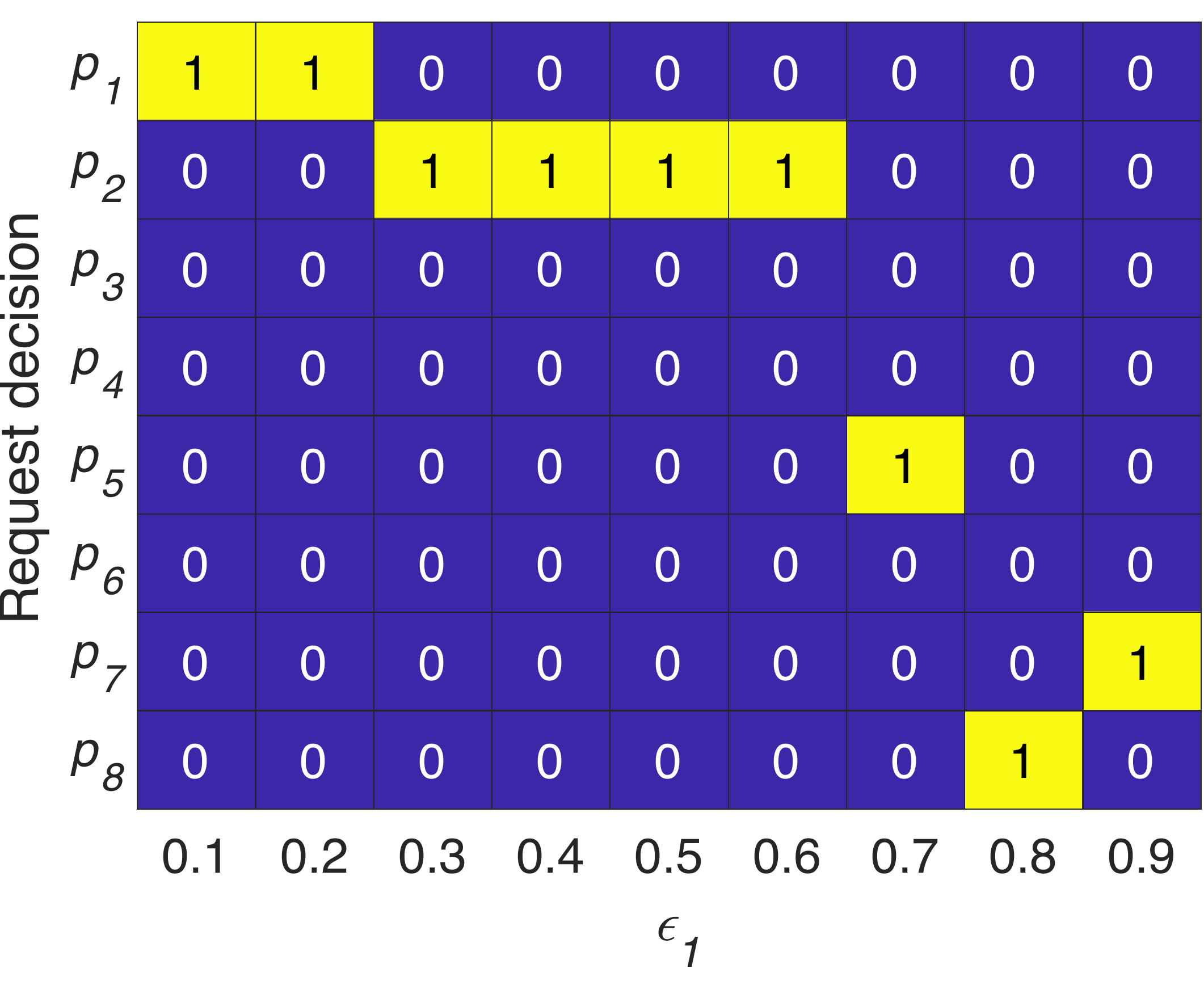}
\end{minipage}
}
\subfigure[$\mathbf{\epsilon}_{-1} = 0.8$.]{
\begin{minipage}[t]{0.45\linewidth}
\centering
\includegraphics[width=1\textwidth]{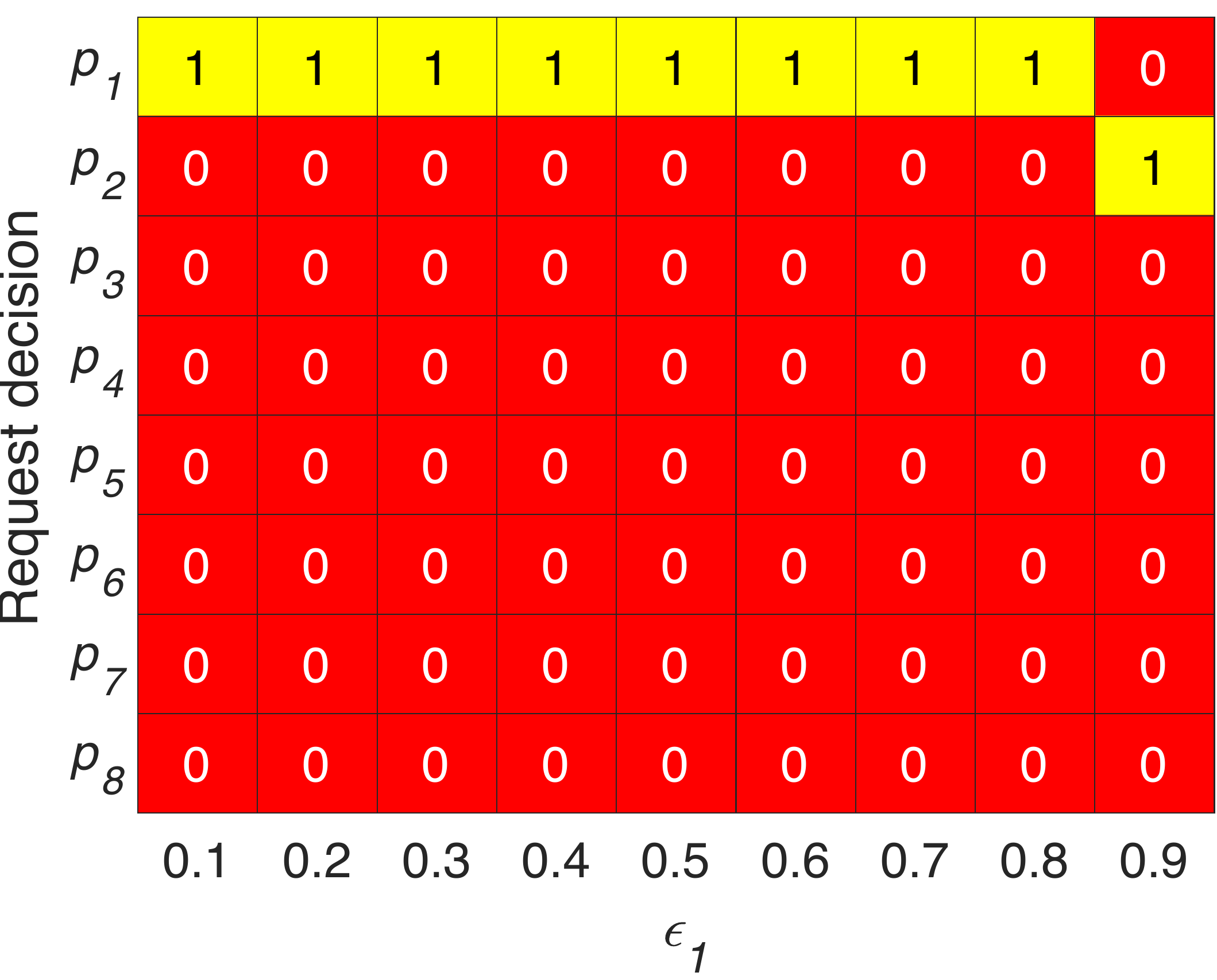}
\end{minipage}
}
\centering
\caption{Synchronization request decision of the partial node.}
\label{fig:decision}
\end{figure}

Finally, we examine the maximized total utility in the NS problem and evaluate its performance under the impacts of the scalar parameter $\alpha_i$ and the cost $C_i$. Here the number of full nodes is still set to $m=8$. The experimental results are reported in Fig. \ref{fig:utility}. From Fig. \ref{fig:utility}(a), one can see that the maximized total utility keeps the same as zero until $\alpha_i=5$, which is because we set $C_i=5$ in this experiment and the profit $\phi_i \in [0,1]$. This means that only when the profit parameter $\alpha_i$ is larger than the cost, can the partial node obtain a positive overall utility. 
While within Fig. \ref{fig:utility}(b), it is shown that the maximized utility first increases with $C_i$ and then decreases when $C_i$ is too large. This is because with a lower $C_i$, the partial node can still obtain a better utility via strategically making the request decision; while when the cost is too high, even the best decision cannot compensate the high resource consumption in request sending process.

\begin{figure}[h]
\centering
\subfigure[Impact of $\alpha_i$.]{
\begin{minipage}[t]{0.45\linewidth}
\centering
\includegraphics[width=1\textwidth]{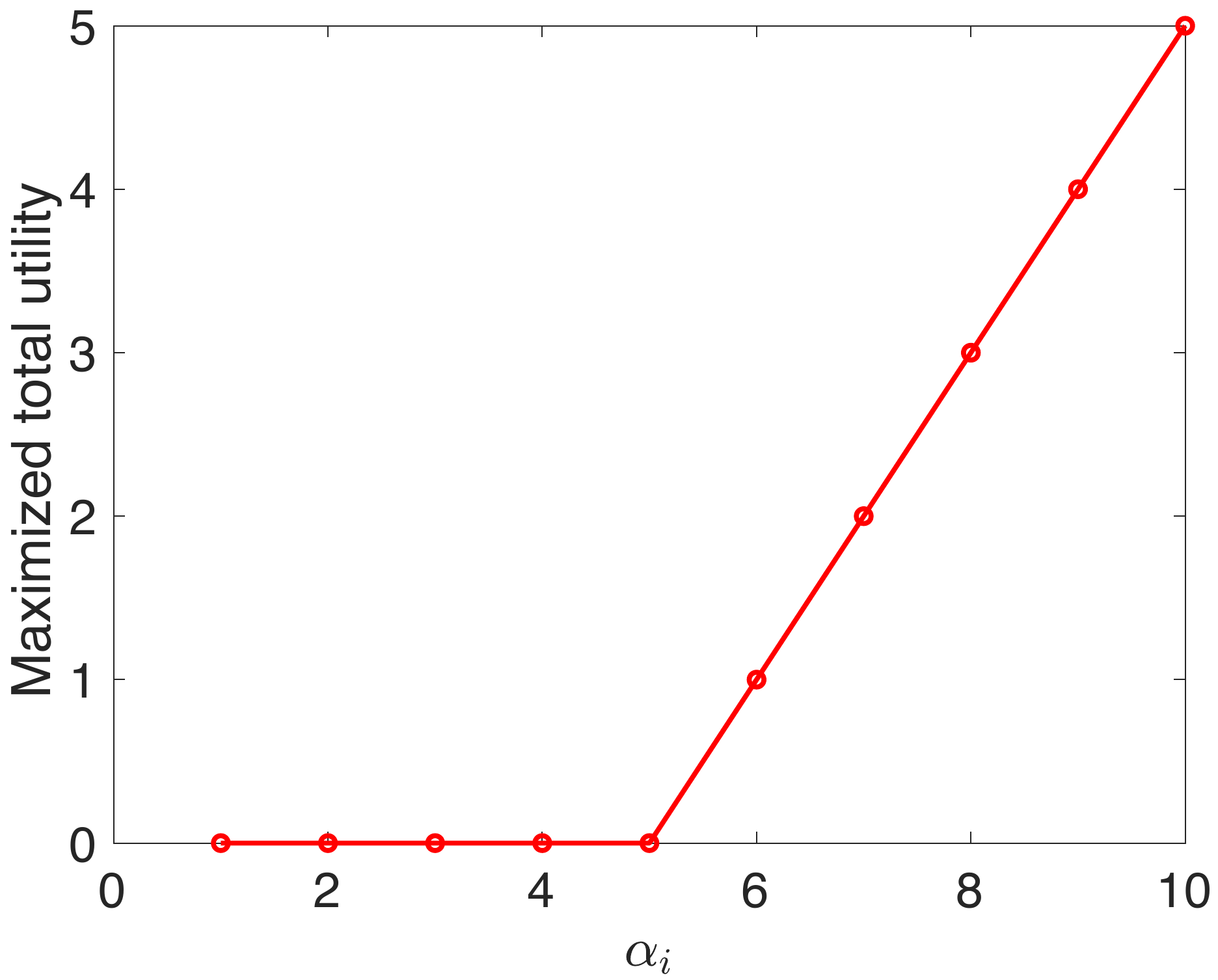}
\end{minipage}
}
\subfigure[Impact of  $C_i$.]{
\begin{minipage}[t]{0.45\linewidth}
\centering
\includegraphics[width=1\textwidth]{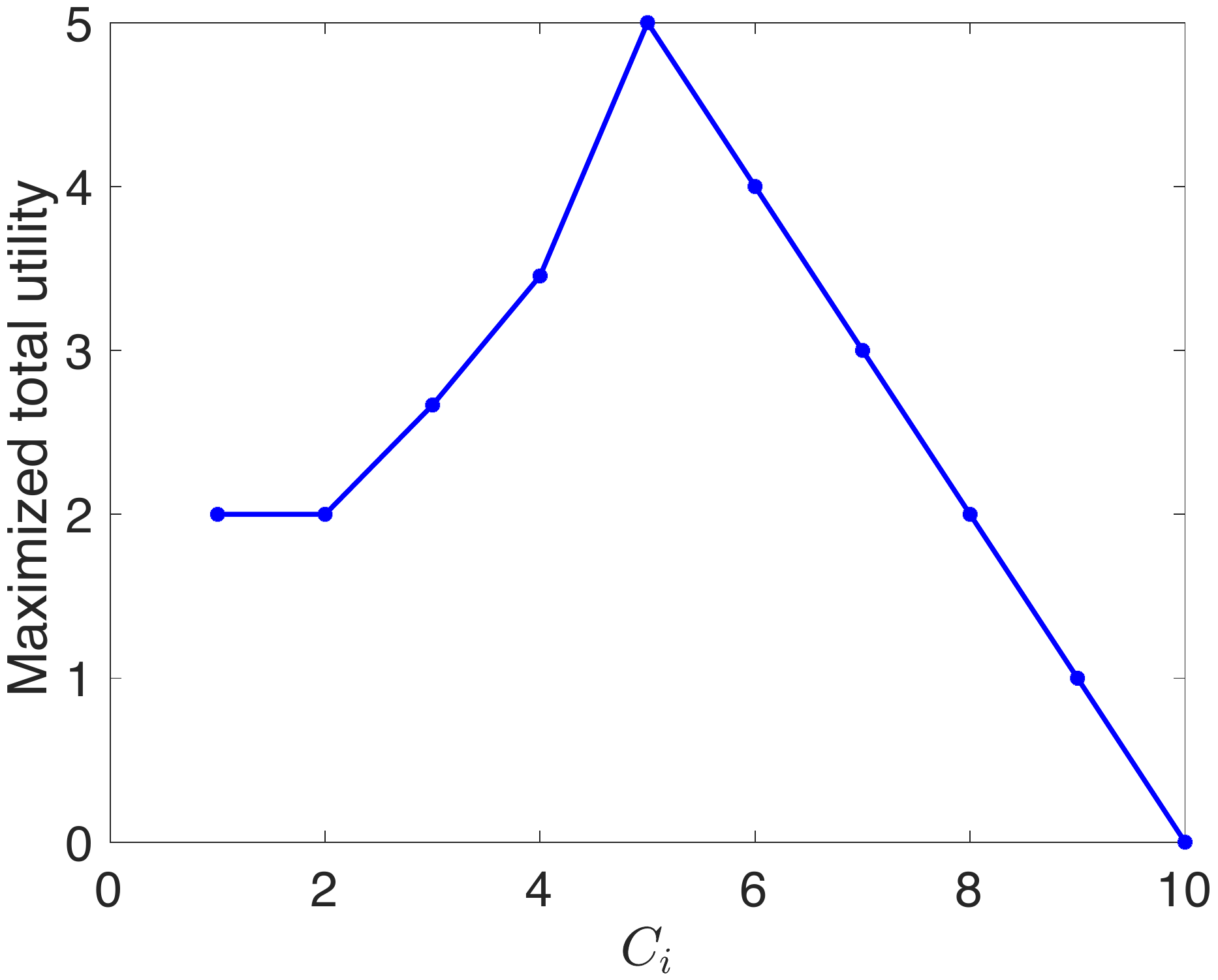}
\end{minipage}
}
\centering
\caption{Maximized total utility of the partial node changing with $\alpha_i$ and $C_i$.}
\label{fig:utility}
\end{figure}

\section{Conclusion}\label{sec:conclusion}
In this paper, we take advantage of the large deviation theory and game theory to study the blockchain synchronization in the pull-based propagation from a microscopic perspective. To be specific, we investigate the fine-grained impacts of peering nodes in synchronizing block information at the node level. On one hand, the full node as the synchronization responder is analyzed based on the queuing model, which reveals the most efficient path to speed up synchronization via increasing the response capacity. On the other hand, the partial node is inspected as the requester, where the best synchronization request scheme is designed using the concept of correlated equilibrium. Extensive experiments are conducted to demonstrate the effectiveness of our analysis.

\bibliographystyle{unsrt}
\bibliography{reference}

\begin{thebibliography}{10}

\bibitem{nakamoto2019bitcoin}
Satoshi Nakamoto.
\newblock Bitcoin: A peer-to-peer electronic cash system.
\newblock Technical report, Manubot, 2019.

\bibitem{market}
Blockchain market.
\newblock
  \url{https://www.marketsandmarkets.com/Market-Reports/blockchain-technology-market-90100890.html}.
\newblock Accessed: 2020-05-30.

\bibitem{gervais2016security}
Arthur Gervais, Ghassan~O Karame, Karl W{\"u}st, Vasileios Glykantzis, Hubert
  Ritzdorf, and Srdjan Capkun.
\newblock On the security and performance of proof of work blockchains.
\newblock In {\em Proceedings of the 2016 ACM SIGSAC conference on computer and
  communications security}, pages 3--16, 2016.

\bibitem{decker2013information}
Christian Decker and Roger Wattenhofer.
\newblock Information propagation in the bitcoin network.
\newblock In {\em IEEE P2P 2013 Proceedings}, pages 1--10. IEEE, 2013.

\bibitem{sompolinsky2015secure}
Yonatan Sompolinsky and Aviv Zohar.
\newblock Secure high-rate transaction processing in bitcoin.
\newblock In {\em International Conference on Financial Cryptography and Data
  Security}, pages 507--527. Springer, 2015.

\bibitem{gobel2016bitcoin}
Johannes G{\"o}bel, Holger~Paul Keeler, Anthony~E Krzesinski, and Peter~G
  Taylor.
\newblock Bitcoin blockchain dynamics: The selfish-mine strategy in the
  presence of propagation delay.
\newblock {\em Performance Evaluation}, 104:23--41, 2016.

\bibitem{kang2018incentivizing}
Jiawen Kang, Zehui Xiong, Dusit Niyato, Ping Wang, Dongdong Ye, and Dong~In
  Kim.
\newblock Incentivizing consensus propagation in proof-of-stake based
  consortium blockchain networks.
\newblock {\em IEEE Wireless Communications Letters}, 8(1):157--160, 2018.

\bibitem{garay2015bitcoin}
Juan Garay, Aggelos Kiayias, and Nikos Leonardos.
\newblock The bitcoin backbone protocol: Analysis and applications.
\newblock In {\em Annual International Conference on the Theory and
  Applications of Cryptographic Techniques}, pages 281--310. Springer, 2015.

\bibitem{pass2017analysis}
Rafael Pass, Lior Seeman, and Abhi Shelat.
\newblock Analysis of the blockchain protocol in asynchronous networks.
\newblock In {\em Annual International Conference on the Theory and
  Applications of Cryptographic Techniques}, pages 643--673. Springer, 2017.

\bibitem{BitcoinRelay}
Bicoin relay network.
\newblock \url{https://github.com/bitcoinfibre/bitcoinfibre}.
\newblock Accessed: 2020-05-30.

\bibitem{klarman2018bloxroute}
Uri Klarman, Soumya Basu, Aleksandar Kuzmanovic, and Emin~G{\"u}n Sirer.
\newblock bloxroute: A scalable trustless blockchain distribution network
  whitepaper.
\newblock {\em IEEE Internet Things J.}, 2018.

\bibitem{conley2018geeq}
John~P Conley.
\newblock The geeq project white paper.
\newblock 2018.

\bibitem{wust2016ethereum}
Karl W{\"u}st and Arthur Gervais.
\newblock Ethereum eclipse attacks.
\newblock Technical report, ETH Zurich, 2016.

\bibitem{cao2001internet}
Jin Cao, W~Cleveland, Dong Lin, and D~Sun.
\newblock Internet traffic tends to poisson and independent as the load
  increases.
\newblock Technical report, Technical report, Bell Labs, 2001.

\bibitem{ganesh2004big}
Ayalvadi~J Ganesh, Neil O'Connell, and Damon~J Wischik.
\newblock {\em Big queues}.
\newblock Springer, 2004.

\bibitem{wang2019corking}
Shengling Wang, Chenyu Wang, and Qin Hu.
\newblock Corking by forking: Vulnerability analysis of blockchain.
\newblock In {\em IEEE INFOCOM 2019-IEEE Conference on Computer
  Communications}, pages 829--837. IEEE, 2019.

\end{thebibliography}

\end{document}